\documentclass[runningheads]{llncs}
\usepackage{graphicx}
\usepackage[ruled]{algorithm}
\usepackage[noend]{algpseudocode}
\usepackage{amsmath}

\usepackage{amssymb}
\usepackage{color}
\usepackage{bbding}
\DeclareMathOperator*{\argmax}{argmax}
\newtheorem{mechanism}{Mechanism}

\begin{document}

\title{ABSNFT: Securitization and Repurchase Scheme for Non-Fungible Tokens Based on Game Theoretical Analysis}
\titlerunning{ABSNFT}

\author{Hongyin Chen\inst{1}\and Yukun Cheng \inst{2} \Envelope\and Xiaotie Deng \inst{1} \Envelope \and Wenhan Huang\inst{3} \and Linxuan Rong\inst{4}} 
\authorrunning{Hongyin Chen et al.}
\institute{Center on Frontiers of Computing Studies, Peking University, Beijing, China \email{\{chenhongyin,xiaotie\}@pku.edu.cn} \and Suzhou University of Science and Technology, Suzhou, China  \\ \email{ykcheng@amss.ac.cn} \and Department of Computer Science, Shanghai Jiao Tong Univerisity, Shanghai, China \\ \email{rowdark@sjtu.edu.cn} \and Washington University in St. Louis, St. Louis, MO, the U.S. \\ \email{l.rong@wustl.edu}}

\maketitle              
\begin{abstract}
The Non-Fungible Token (NFT) is viewed as one of the important applications of blockchain technology. Currently NFT has a large market scale and multiple practical standards, however several limitations of the existing mechanism in NFT markets still exist. This work proposes a novel securitization and repurchase scheme for NFT to overcome these limitations. We first provide an Asset-Backed Securities (ABS) solution to settle the limitations of non-fungibility of NFT. Our securitization design aims to enhance the liquidity of NFTs and enable Oracles and Automatic Market Makers (AMMs) for NFTs. Then we propose a novel repurchase protocol for a participant owing a portion of NFT to repurchase other shares to obtain the complete ownership. As the participants may strategically bid during the acquisition process, we formulate the repurchase process as a Stackelberg game to 
explore the equilibrium prices. We also provide solutions to handle difficulties at market such as budget constraints and lazy bidders.
\end{abstract}

\section{Introduction}\label{Introduction}
Ever since the birth of the first piece of Non-Fungible Token (NFT)~\cite{fairfield2021tokenized}~\cite{wang2021non}, the world has witnessed an extraordinarily fast growth of its popularity. NFT markets, especially Opensea\footnote{Opensea Platform. \url{https://opensea.io/}}, have prospered with glamorous statistics of a total of over 80 million pieces of NFTs on the platform and a total transaction volume of over 10 billion US dollars.\footnote{Data source from Opensea \url{https://opensea.io/about}}

NFT is a type of cryptocurrency that each token is non-fungible. 
The first standard of NFT, ERC-721~\cite{erc721}, gives support to a type of tokens that each has a unique identifier. The feature of uniqueness makes NFTs usually be tied to specific assets, such as digital artwork and electronic pets. Some researches also explore the application of NFT in patent, copyright and physical assets~\cite{ccauglayan2021nfts}~\cite{valeonti2021crypto}.

The technology of NFT has also advanced rapidly. Besides ERC-721, ERC-1155~\cite{erc1155} is also a popular standard of NFT. ERC-1155 is a flexible standard that supports multiple series of tokens, each series is a type of NFT or Fungible Token (FT).
NFT protocols are usually derived by smart contracts in a permissionless blockchain, but there are now some NFT designs for permissioned blockchains~\cite{hong2019design}.

Although NFT has a large market scale and multiple practical standards, there still exist several limitations in NFT market, one of which is the poor liquidity.
 
The issue of liquidity is crucial in both De-fi and traditional finance. Usually, if assets have higher liquidity, they would have higher trading volume, and further have higher prices~\cite{amihud1991liquidity}.
Particularly, in blockchain, the liquidity of Fungible Tokens, such as wBTC and ETH, has been enhanced by Oracles~\cite{mammadzada2020blockchain} and Automated Market Makers (AMMs)~\cite{angeris2020improved} like Uniswap and Sushiswap. However, the non-fungibility property of NFT leads to poor liquidity. For this reason, the existing NFT marketplace usually uses the English Auction or Dutch Auction to trade NFTs~\cite{kong2021alternative}.

\begin{itemize}
    \item Firstly, Non-Fungibility means indivisible.

    As the NFT series with the highest market value, CryptoPunks has an average trading price of 189 Eth\footnote{90-day average before November 22, 2021. Data source from \url{https://opensea.io/activity/cryptopunks}}, which is worth more than 790, 000 U.S. dollars\footnote{The price of Eth here refers to the data on November 22, 2021. \url{https://etherscan.io/chart/etherprice}}. If bitcoins are expensive, we can trade 0.01 bitcoins, but we can’t trade 0.01 CryptoPunks. As a result, the liquidity of CryptoPunks is significantly lower than other NFT series. Therefore, the liquidity for NFTs with high values is limited. 
    
    \item Secondly, shared ownership is not allowed because of Non-Fungibility. Therefore, it’s difficult to reduce risk and enhance the liquidity of NFTs through portfolios. What’s more, Some NFT assets such as patents need financial support to foster the process of development. They would require a means to attract finance. The above two limitations also exist in traditional settings.
    
    \item Thirdly, the feature of non-fungible makes NFT unable to be directly applied in Oracles~\cite{mammadzada2020blockchain} and Automated Market Makers (AMMs)~\cite{angeris2020improved}, which are important methods of pricing in the blockchain. This is because fungibility is the basis of Oracles and AMMs.
\end{itemize}

\subsection{Main Contributions}
We present ABSNFT, a securitization and repurchase scheme for NFT, which overcomes the above-mentioned limitations from the following three aspects.

\begin{itemize}
    \item Firstly, we propose an Asset-Backed Securities (ABS)~\cite{bhattacharya1996asset} solution to settle the limitations of non-fungibility of NFT. We design a smart contract including three parts: NFT Securitization Process, NFT Repurchase Process, and NFT Restruction Process. In our smart contract, a complete NFT can be securitized into fungible securities, and fungible securities can be reconstructed into a complete NFT.
    
    The securitization process manages to resolve the majority of issues the current NFT application is confronted with: the securities of NFT have lower values compared to the complete one before securitization, which increases market liquidity; securities could act as fungible tokens that can be applied in Oracles and AMMs; the investment risk is being reduced dramatically; financing is possible since securities can belong to different owners.
    
    \item Secondly, we design a novel repurchase process based on Stackelberg game~\cite{von2010market}, which provides a mechanism to repurchase NFT securities at a fair price. The NFT Repurchase Process can be triggered by the participant who owns more than half of the securities of the NFT. We analyze the Stackelberg Equilibrium (SE) in three different settings and get good theoretical results. 

    \item Thirdly, we propose solutions to the budget constraints and lazy bidders, which make good use of the decentralization of blockchain. We propose a protocol that allows participants to accept financial support in the repurchase game to reduce the influence of budget constraints. We also propose two solutions for players that might not bid in the game, which prevent the game process from being blocked and protect the utility of lazy bidders.
\end{itemize}

\subsection{Related Works}\label{related}
In financial research, there are two well-studied repurchase scenarios, repurchase agreement and stock repurchase.

Repurchase agreement is a short-term transaction between two parties in which one party borrows cash from the other by pledging a financial security as collateral~\cite{acharya2011repurchase}. The former party is called the security issuer, and the latter party is called the investors. To avoid the failure of liquidation, the security issuer needs to mortgage assets or credit. An instance of such work from the Federal Reserve Bank of New York Quarterly Review introduces and analyzes a repurchase agreement for federal funds~\cite{lucas1977federal}. The Quarterly Review describes the repurchase agreement as “involving little risk”, as either parties’ interests are been safeguarded.

Studies of repurchase agreement cannot be directly applied to our topic. The key point is that the problem we are studying is not to mortgage NFTs to obtain cash flow, but to securitize NFTs to overcome the restrictions of non-fungibility. What's more, the repurchase prices are usually derived from the market model. But the NFT market is not as mature as the financial market, which makes it hard to calculate a fair price through the market model.

Stock repurchase refers to the behavior that listed companies repurchase stocks from the stockholders at a certain price~\cite{constantinides1989optimal}. Usually, stock repurchase is adopted to release positive signals to the stock market and doesn't aim to repurchase all stocks. However, NFTs usually need to be complete without securities in cross-chain scenarios.

Oxygen~\cite{oxygen2018} is a decentralized platform that supports repurchase agreement based on digital assets. In Oxygen, users can borrow cash flow or assets with good liquidity by pledging assets with poor liquidity. The repurchase prices and the evaluations of assets are provided by a decentralized exchange, Serum~\cite{serum}. However, such pricing method is dangerous because decentralized exchanges are very vulnerable to attacks like flash loans~\cite{wang2021towards}.

ABSNFT is distinguished among all these works because it adapts well to the particularities of NFT market and blockchain.

\begin{itemize}
    \item First, the securities in ABSNFT represent property rights rather than creditor's rights. Investors do not need to worry that the cash flow or the mortgaged assets of the securities issuer may not cover the liquidation, which may be risky in a repurchase agreement. What's more, any investor can trigger a repurchase process as long as he owns more than half of the shares.
    \item Second, the repurchase process of ABSNFT doesn't depend on market models or exchanges. The repurchase price is decided by the bids given by participants, and every participant won't get negative utility if he bids truthfully.
    \item Third, ABSNFT has well utilized the benefits of blockchain technology. The tradings of securities are driven by the smart contract. The operations of ABSNFT don't rely on centralized third-party and are available $24\times 7$ for participants.
\end{itemize}

The rest of the paper is arranged as follows. Section~\ref{NFTScheme} introduces the NFT securitization process. In Section~\ref{2pGame} and Section~\ref{2PRG}, we study the two-player repurchase game in a single round and the repeated setting. In Section~\ref{MPGame}, we analyze the repurchase game with multiple leaders and one follower. In section~\ref{discussion}, we discuss the solution to the issues with budget constraints and lazy bidders in the blockchain setting. In the last section, we give a summary of ABSNFT and propose some future works.

\section{NFT Securitization and Repurchase Scheme}\label{NFTScheme}

In this section, we would like to introduce the general framework of the smart contract, denoted by $C_{NFT}$, which includes the securitization process, the trading process, repurchase process and restruction process for a given NFT.

=
\subsection{Basic Setting of NFT Smart Contract} There are two kinds of NFTs discussed in this paper.
\begin{itemize}
    \item \textbf{Complete NFT}. Complete NFTs are conventional non-fungible tokens, which appear in blockchain systems as a whole. Each complete NFT has a unique token ID. We use $CNFT(id)$ to denote one complete NFT with token ID of $id$.

    \item \textbf{Securitized NFT}. Securitized NFTs are the $Asset~Based~Securities$ of complete NFTs. A complete NFT may be securitized into an amount of securitized units. All units of securitized NFTs from a complete NFT $CNFT(id)$ have the same ID, associated with the ID of $CNFT(id)$.
    Thus we denote the securitized NFT by $SNFT(id)$. In our smart contract, all securitized NFTs can be freely traded.
\end{itemize}

In our setting, all complete NFTs and securitized NFTs belong to one smart contract, denoted by $C_{NFT}$. Although the securitized NFTs are similar to the fungible tokens in ERC-1155 standard~\cite{erc1155}, our smart contract $C_{NFT}$ is actually quite different from ERC-1155 standard. That is because all securitized NFTs in $C_{NFT}$, associated to one complete NFT, have the same ID,
while different NFTs and different fungible tokens generally have different token IDs in ERC-1155 standard. Therefore, our $C_{NFT}$ is based on ERC-721 standard~\cite{erc721}, and the complete NFTs are just the NFTs defined in ERC-721. Table \ref{tab1} lists all functions in $C_{NFT}$.

\begin{table}
\caption{The key functions of $C_{NFT}$}\label{tab1}
\begin{tabular}{|p{4cm}|p{8cm}|}
\hline
{\bf Function Name} & {\bf Function Utility} \\
\hline
 $CNFTownerOf(id)$ & Return the address of the owner of $CNFT(id)$.\\
 \hline
 $CNFTtransferFrom$ $(addr1, addr2, id)$ & Transfer the ownership of $CNFT(id)$
 from address $addr1$ to address $addr2$. Only the owner of $CNFT(id)$ has the right to trigger this function.\\
 \hline
 $SNFTtotalSupply(id)$ & Return the total amount of $SNFT(id)$ in contract $C_{NFT}$.\\
 \hline
 $SNFTbalanceOf(addr, id)$ & Return the amount of $SNFT(id)$ owned by address $addr$.\\
 \hline
 $SNFTtransferFrom$ $($addr1$, $addr2$, id, amount)$ & Transfer the ownership of $amout$ unit of $SNFT(id)$ from address $addr1$ to address $addr2$. \\
 \hline
 $CNFTsecuritization$ $(addr, id, amount)$ & Freeze $CNFT(id)$, and then transfer $amout$ units of $SNFT(id)$ to address $addr$. Only the owner of $CNFT(id)$ can trigger this function.\\
 \hline
 $CNFTrestruction(addr, id)$ & Burn all $SNFT(id)$, unfreeze $CNFT(id)$, and then transfer the ownership of $CNFT(id)$ to address $addr$. Only the one who owns all amounts of  $SNFT(id)$ can trigger this function.\\
 \hline
 $Repurchase(id)$& Start the repurchase process of $SNFT(id)$. Only the one who owns more than half amounts of $SNFT(id)$ can trigger this function.\\
 \hline
\end{tabular}
\end{table}

The task of smart contract $C_{NFT}$ includes securitizing complete NFTs, trading the securitized NFTs among participants, and restructing complete NFT after repurchasing all securitized NFTs with the same ID. Bescause the transactions of securitized NFTs are similar to those of fungible tokens, we omit the trading process here and introduce NFT securitization process, NFT repurchse process and NFT restruction process in the subsequent three subsections respectively.

\subsection{NFT Securitization Process}
This subsection focuses on the issue of Asset-Backed Securities for Complete NFTs. We propose Algorithm \ref{CNFTsecuritization} to demonstrate the NFT securitization process. To be specific, once $CNFTsecuritization(addr,id,amount)$ is triggered by the owner of $CNFT(id)$, the $amount$ units of securitized NFTs are generated and transferred to address $addr$ in Line 2-4; then the ownership of $CNFT(id)$ would be transferred to a fixed address $FrozenAddr$ in Line 5.  

It is worth to note that if $Repurchase(id)$ has not been triggered, securitized NFTs can be freely traded in blockchain system.

\begin{algorithm}
    \caption{NFT Securitization}
    \label{CNFTsecuritization}
    \algrenewcommand\algorithmicwhile{\textbf{upon}}
    \begin{algorithmic}[1]
        \Procedure{CNFTsecuritization}{}\Comment{Triggered by $sender$}

            \State $require(sender == CNFTownerOf(id))$ \Comment{$sender$ is the owner of $CNFT(id)$}
            \State $totalSupply[id] \gets amount$ \Comment{Record the total amount of units of $SNFT(id)$}
            \State $tokenBalance[id][addr] \gets amount$ \Comment{the $amount$ units $SNFT(id)$ are generated and transferred to address $addr$}
            \State $CNFTtransferFrom(sender, FrozenAddr, id)$ \Comment{Freeze $CNFT(id)$}
        \EndProcedure
    \end{algorithmic}
\end{algorithm}

\subsection{NFT Repurchase Process}

After the securitization process, a complete NFT $CNFT(id)$ is securitized into $M$ units of $SNFT(id)$. Suppose that there are $k+1$ participants, $N=\{N_0,\cdots,N_{k}\}$, each owning $m_i$ units of $SNFT(id)$. Thus $\sum_{i=0}^{k}m_i=M$. If there is one participant, denoted by $N_0$, owing
more than half of $SNFT(id)$ (i.e. $m_0>\frac12 M$), then he can trigger the repurchase process by trading with each $N_i$, $i=1,\cdots,k$. Majority is a natural requirement for a participant to trigger a repurchase mechanism, and thus our repurchase mechanism sets the threshold as $\frac12$. In addition, if the trigger condition is satisfied (i.e., someone holds more than half of shares), then there must be exactly one participant who can trigger the repurchase mechanism. This makes our mechanism easy to implement. Our mechanism also works well if the threshold is larger than $\frac12$.

Let $v_i$ be $N_i$'s value estimate for one unit of $SNFT(id)$ and $p_i$ be the bid provided by $N_i$, $i=0,\cdots,k$, in a deal. Specially, our smart contract $C_{NFT}$ requires each value $v_i\in \{1,\cdots\}$ and bid $p_i\in \{0,1,\cdots\}$ to discretize our analysis. 
We assume that the estimation of $v_i$ is private information of $N_i$, not known to others. The main reason is that most of NFT objects, such as digital art pieces, would be appreciated differently in different eyes. 

Participants may have different opinions about a same NFT, which makes each of them has a private value $v_i$. Without loss of generality, we assume that $N_i$'s private value on the complete NFT is $M\cdot v_i$.

\begin{mechanism}\label{mechanism1}
	{\bf (Repurchase Mechanism)} Suppose participant $N_0$ owes more than half of $SNFT(id)$ and triggers the repurchase mechanism.
 For the repurchase between $N_0$ and $N_i$, $i=1,\cdots,k$,
\begin{itemize}
	\item if $p_0\geq p_i$, then $N_0$ successfully repurchases $m_i$ units of $SNFT(id)$ from $N_i$ at the unit price of $\frac{p_0+p_i}{2}$;
	\item if $p_0\leq p_i-1$, then $N_0$ fails to repurchase, and then he shall sell $m_i$ units of $SNFT(id)$ to $N_i$. The unit price that $N_i$ pays is $\frac{p_0+p_i}{2}$, and $N_0$ obtains a discounted revenue $\frac{p_0+p_i-1}{2}$ for each unit of $SNFT(id)$.
\end{itemize}
\end{mechanism}

Mechanism 1 requires that the repurchase process only happens between $N_0$ and $N_i$, $i=1,\cdots,k$. 
If $p_0\geq p_i$, then $N_0$ successfully repurchases $m_i$ units of $SNFT(id)$ from $N_i$, and the utilities of $N_0$ and $N_i$ are

\begin{eqnarray}\label{uN0}
  U_0^i(p_0,p_i)=m_i(v_0 -\frac{p_0 + p_i}{2} ),~U_i(p_0,p_i)=m_i(\frac{p_0 + p_i}{2}-v_i),~if \;p_0 \geq p_i.
\end{eqnarray}
If $p_0<p_i$, then $N_0$ fails to repurchase from $N_i$, and the utilities of $N_0$ and $N_i$ are
 \begin{eqnarray}\label{uNi}
  U_0^i(p_0,p_i)=m_i(\frac{p_0 + p_i-1}{2}-v_0 ),~U_i(p_0,p_i)=m_i(v_i-\frac{p_0 + p_i}{2}),~if \;p_0 \leq p_i-1.
\end{eqnarray}

All participants must propose their bids rationally under Mechanism 1. If the bid $p_0$ is too low, $N_0$ would face the risk of repurchase failure. Thus, the securities of $N_0$ would be purchased by other participants at a low price, and $N_0$'s utility may be negative. Similarly, if bid $p_i$ of $N_i$, $i=1,\cdots,k$, is too high, $N_i$ would purchase securities with an extra high price and get negative utility. However, if a participant bids truthfully, he always obtains non-negative utility.

During the repurchase process, the key issue for each participant is how to bid  $p_i$, $i=0,\cdots,k$, based on its own value estimation. To solve this issue, we would model the repurchase process as a stackelberg game to explore the equilibrium pricing solution in the following Section 3 to 5.

\subsection{NFT Restruction Process}
Once one participant successfully repurchases all securitized NFTs, he has the right to trigger $CNFTrestruction(addr, id)$, shown in Algorithm \ref{CNFTrestruction}, to burn these securitized NFTs in Line 3 to 4 and unfreeze $CNFT(id)$, such that  the ownership of $CNFT(id)$ would be transferred from address $FrozenAddr$ to this participant's address $addr$ in Line 5.

After NFT restruction, all $SNFT(id)$ are burnt, and $CNFT(id)$ is unfrozen. The owner of $CNFT(id)$ has the right to securitize it or trade it as a whole.

\begin{algorithm}
    \caption{NFT Restruction}
    \label{CNFTrestruction}
    \algrenewcommand\algorithmicwhile{\textbf{upon}}
    \begin{algorithmic}[1]
        \Procedure{CNFTrestruction}{}\Comment{Triggered by $sender$}

            \State $require(tokenBalance[id][sender] == totalSupply[id])$ \Comment{$sender$ should be the owner of all $SNFT(id)$}
            \State $totalSupply[id] \gets 0$ \Comment{Burn all $SNFT(id)$}
            \State $tokenBalance[id][sender] \gets 0$ \Comment{Burn all $SNFT(id)$}
            \State $CNFTtransferFrom(FrozenAddr, addr, id)$ \Comment{Unfreeze $CNFT(id)$}
        \EndProcedure
    \end{algorithmic}
\end{algorithm}

\section{Two-Player Repurchase Stackelberg Game}\label{2pGame}

This section discusses the repurchase process for a two-player scenario. To be specific, in the two-player scenario, when a player owns more than half of $SNFT(id)$, denoted by $N_0$, he will trigger the repurchase process with another player $N_1$.
To explore the optimal bidding strategy for both players, we model the repurchase process as a two-stage Stackelberg game, in which $N_1$ acts as the leader to set its bid $p_1$ in Stage I, and $N_0$, as the follower, decides its bid $p_0$ in Stage II. Recall that all bids and all values are in $\{0,1,\cdots\}$.
\begin{enumerate}
  \item[(1)] \emph{$N_0$'s bidding strategy in Stage II:} Given the bid of $p_1$, set by $N_1$ in Stage I, $N_0$ decides its bid to maximize its utility, which is given as:
\begin{equation}\label{Utility_of_N0}
   U_0(p_0,p_1)=
 \begin{cases}
  m_1(v_0 -\frac{p_0 + p_1}{2} )  & if \;p_0 \geq p_1; \\
  m_1(\frac{p_0 + p_1-1}{2} -v_0)  & if \;p_0 \leq p_1-1.
 \end{cases}
  \end{equation}
 \item[(2)] \emph{$N_1$'s bidding strategy in Stage I:} Once obtain the optimal bid $p_0^*(p_1)$ of $N_0$ in Stage II, which is dependent on $p_1$, $N_1$ goes to compute the optimal bid $p_1^*$ by maximizing his utility function $max_{p_1}U_1(p^*_0(p_1),p_1)$, where 
\begin{equation}\label{Utility_of_N1}
   U_1(p_0,p_1)=
 \begin{cases}
  m_1(\frac{p_0 + p_1}{2}-v_1)  & if \;p_1\leq p_0; \\
  m_1(v_1-\frac{p_0 + p_1}{2})  & if \;p_1\geq p_0+1.
 \end{cases}
  \end{equation}
\end{enumerate}

\subsection{Analysis under Complete Information}

(1) {\bf Best response of $N_0$ in Stage II.} Given the bid $p_1$ provided by $N_1$, in Stage II, $N_0$ shall determine the best response $p_0^*(p_1)$ to maximize his utility.

\begin{lemma}\label{lemma1}
In the two-stage Stackelberg game for repurchase process, if the bid $p_1$ is given in Stage I, the best response of $N_0$ in Stage II is
\begin{equation}
    p_0^*(p_1)=
    \begin{cases}
    p_1 - 1 & if\;p_1\geq v_0+1 \\
    p_1  & if \;p_1 \leq v_0
    \end{cases}
\end{equation}
\end{lemma}

\begin{proof}
According to (\ref{Utility_of_N0}), $U_0$ is monotonically increasing when $p_0 \leq p_1-1$ and monotonically decreasing when $p_0 \geq p_1$. So $p_0^*(p_1)\in\{p_1-1,p_1\}$. In addition, when $p_1\geq v_0+1$, we have
\begin{eqnarray*}
  U_0(p_0=p_1,p_1)=m_1(v_0-p_1)<0\leq m_1(p_1-1-v_0)=U_0(p_0=p_1-1,p_1).
\end{eqnarray*}
It implies that the best response of $N_0$ is $p_0^*(p_1)=p_1-1$ if $p_1\geq v_0+1$.
When $p_1\leq v_0$, we have
\begin{eqnarray*}
  U_0(p_0=p_1,p_1)=m_1(v_0-p_1)\geq 0>m_1(p_1-1-v_0)=U_0(p_0=p_1-1,p_1).
\end{eqnarray*}
So under the situation of $p_0\leq v_0$, the best response of $N_0$ is $p_0^*(p_1)=p_1$. \qed
\end{proof}

\noindent(2) {\bf The optimal strategy of $N_1$ in Stage I.} The leader $N_1$ would like to optimize his bidding strategy to maximize his utility shown in (\ref{Utility_of_N1}).

\begin{lemma}\label{lemma2}
In the two-stage Stackelberg game for repurchase process, the optimal bidding strategy for the leader $N_1$ is
\begin{equation}
   p_1^*=
    \begin{cases}
    v_0 & if\;v_1 \leq v_0 \\
    v_0 + 1  & if \; v_1 \geq v_0+1.
    \end{cases}
\end{equation}
\end{lemma}
\begin{proof}
Based on Lemma~\ref{lemma1}, we have
\begin{equation*}\label{2PU1}
   U_1(p_0^*(p_1),p_1)=\\
   \begin{cases}
    m_1(p_1 - v_1) & if \;p_1 \leq v_0; \\
    m_1(v_1 - p_1 + \frac{1}{2})& if \;p_1 \geq  v_0+1.
   \end{cases}
\end{equation*}
Thus $U_1$ is monotonically increasing when $p_1 \leq v_0$ and monotonically decreasing when $p_1 \geq v_0+1$, indicating the optimal bidding strategy $p_1^*\in\{v_0,v_0+1\}$. In addition, for the case of $v_0\geq v_1$, if $p_1=v_0$, then $p_0^*(p_1)=p_1=v_0$ by Lemma \ref{lemma1} and
$
  U_1(v_0,v_0)=m_1(v_0-v_1)\geq 0.
$
On the other hand, if $p_1=v_0+1$, then $p_0^*(p_1)=p_1-1=v_0$ by Lemma \ref{lemma1} and
$
  U_1(v_0,v_0+1)=m_1(v_1-v_0-\frac12)< 0.
$
Therefore, $U_1(v_0,v_0)>U_1(v_0,v_0+1)$, showing the optimal bidding strategy of $N_1$ is $p_1^*=v_0$ when $v_0\geq v_1$. Similarly, for the case of $v_0\leq v_1-1$, we can conclude that $p_1^*=v_0+1$. This lemma holds. \qed
\end{proof}
Combining Lemma~\ref{lemma1} and~\ref{lemma2}, the following theorem can be derived directly.
\begin{theorem}\label{t1}
When $v_0 \geq v_1$, there is exactly one Stackelberg equilibrium where $p_1^* = p_0^* = v_0$. And when $v_0 \leq  v_1-1$, there is exactly one Stackelberg equilibrium where $p_0^*= v_0,~p_1^* = v_0 + 1$.
\end{theorem}

Furthermore, the following theorem demonstrates the relation between Stackelberg equilibrium and Nash equilibrium.
\begin{theorem}\label{SEalsoNE}
Each Stackelberg equilibrium in Theorem \ref{t1} is also a Nash equilibrium.
\end{theorem}
The proof of Theorem~\ref{SEalsoNE} is provided in Appendix A.

\subsection{Analysis of Bayesian Stackelberg Equilibrium}\label{2PBSE}

In the previous subsection, the Stackelberg equilibrium is deduced based on the complete information about the value estimate $v_0$ and $v_1$. However, the value estimates may be private in practice, which motivates us to study the Bayesian Stackelberg game with incomplete information. In this proposed game, although the value estimate $v_i$ is not known to others, except for itself $N_i$, $i=0,1$, the probability distribution of each $V_i$ is public to all. Here we use $V_i$ to denote the random variable of value estimate. Based on the assumption that all $V_i$ are integers, we continue to assume that each $N_i$'s value estimate $V_i$ has finite integer states, denoted by $v_i^1,v_i^2,\cdots, v_i^{k_i}$, and its discrete probability distribution is $Pro(V_i=v_i^l)=P_i^l$, $l=1,\cdots,k_i$, and $\sum_{l=1}^{k_i}P_i^l=1$, $i=0,1$.

\noindent(1) {\bf Best response of $N_0$ in Stage II.} Because $v_0$ is deterministic to $N_0$, and $p_1$ is given by $N_1$ in Stage I, Lemma~\ref{lemma1} still holds, so
\begin{equation*}
    p_0^*(p_1)=
    \begin{cases}
    p_1-1 & if\;p_1\geq v_0+1; \\
    p_1  & if \;p_1 \leq v_0.
    \end{cases}
\end{equation*}

\noindent(2) {\bf Optimal bidding strategy of $N_1$ in Stage I.}
By Lemma~\ref{lemma1}, we have
\begin{equation*}
   U_1(p_0^*(p_1),p_1)=\\
   \begin{cases}
    m_1 (p_1 - v_1) & if \;p_1 \leq v_0; \\
    m_1 (v_1 - p_1 + \frac{1}{2})& if \;p_1 \geq v_0+1.
   \end{cases}
\end{equation*}

Based on the probability distribution of $V_0$, the expected utility of $U_1$ is:
\begin{eqnarray}\label{expected utility}
  E_1(p_1)&=\sum_{v_0^l\geq p_1}m_1(p_1-v_1)P_0^l+\sum_{v_0^l\leq p_1-1}m_1(v_1-p_1+\frac12)P_0^l
\end{eqnarray}
To be specific, if $p_1\geq v_0^{k_0}+1$, then $E_1(p_1)=m_1(v_1-p_1+\frac12)$, and $N_1$ obtains his maximal expected utility at $p_1^*=v_0^{k_0}+1$. If $p_1\leq v_0^1$, then $E_1(p_1)=m_1(p_1-v_1)$, and $N_1$ obtains his maximal expected utility at $p_1^*=v_0^{1}$. If there exists an index $l$, such that $v_0^{l-1}<p_1\leq v_0^l$, $l=2,\cdots,k_0$, then
\begin{equation*}
    E_1(p_1)=\sum_{h=1}^{l-1}m_1(v_1-p_1+\frac12)P_0^h+\sum_{h=l}^{k_0}m_1(p_1-v_1)P_0^h.
\end{equation*}
Therefore, $N_1$ can obtain his maximal expected utility at $p_1^*=v_0^l$, when $\sum_{h=l}^{k_0}P_0^h\geq \sum_{h=1}^{l-1}P_0^h$. Otherwise, $N_1$'s maximal expected utility is achieved at $p_1^*=v_0^{l-1}+1$. Hence, the optimal bid $p_1^*\in \{v_0^l,~v_0^l+1\}_{l=1,\cdots,k_0}$.

\begin{theorem}\label{BSE}
There is a Stackelberg equilibrium in the Bayesian Stackelberg game.
\begin{enumerate}
  \item [(1)] If $p_1^*\leq v_0$, then $p_0=p_1^*$ and $p_1=p_1^*$ is a Stackelberg equilibrium.
  \item [(2)] If $p_1^*\geq v_0+1$, then $p_0=p_1^*-1$ and $p_1=p_1^*$ is a Stackelberg equilibrium.
\end{enumerate}
\end{theorem}

\section{Repeated Two-Player Stackelberg Game}\label{2PRG}
This section would extend the study of the one-round Stackelberg game in the previous section to the repeated Stackelberg game. Before our discussion, we construct the basic model of a repeated two-player Stackelberg game by introducing the necessary notations.
\begin{definition}
Repeated two-player Stackelberg repurchase game is given by a tuple $G_r = (M, N,V,S,L,P,U)$, where:
\end{definition}
\begin{itemize}
\item $N = \{N_0,N_1\}$ is the set of two participants. The role of being a leader or a follower may change in the whole repeated process.
\item $M$ is the total amount of $SNFT(id)$. W.l.o.g., we assume that $M$ is odd, such that one of $\{N_0,N_1\}$ must have more than half of $SNFT(id)$.
\item $V = \{v_0,v_1\}$ is the set of participants' value estimates. Let $v_i \in \{1,2,3,\cdots\}$ be an integer.
\item $S = \{s^1,s^2,\cdots, s^t, z\}$ is the set of sequential states. $s^j = (m_0^j,m_1^j)$, in which $m_0^j,m_1^j> 0$ are integers, $m_0^j + m_1^j= M$, and $m_0^j\neq m_1^j$ because $M$ is odd. $z \in Z = \{z_0,z_1\}$ represents the terminal state, where $z_0 = (M,0), z_1 = (0,M)$.  If the sequential states are infinity, then $t = +\infty$. Let us denote $(m_0^{t+1},m_1^{t+1}) = z$.
\item $L = \{l^1,l^2,\cdots,l^t\}$ is the set of sequential leaders, where $l^j$ is the leader in the $j$-th round. To be specific, $l^j=N_1$, if $m_0^j > m_1^j$; otherwise, $l^j=N_0$. It shows the participant who triggers the repurchase process in each round should be the follower.
\item $P_i = \{p_i^1, p_i^2,\cdots,p_i^t\}$ is the set of sequential prices bidded by $N_i$, $p_i^j \in \{0,1,2,\cdots\}$.
\item $U_i:S \times P_0 \times P_1 \xrightarrow[]{}R$ is the utility function of player $N_i$ in a single round. The detailed expressions of $U_i$ will be proposed later.
\end{itemize}
In practice, $v_i$, $i=0,1$, may not be common information. However, we can extract them from the historical interaction data of the repeated game by online learning \cite{weed2016online} or reinforcement learning \cite{li2019latent} methods. Therefore, we mainly discuss the case with complete information in this section.

\noindent{\bf Repeated Stackelberg Game Procedure}
Repeated game $G_r$ consists of several rounds, and each round contains two stages. In the $j$-th round,
\begin{itemize}
    \item In \textbf{Stage I}, the leader provides a bid $p_i^j \in \{0,1,\cdots\}$.
    \item In \textbf{Stage II}, the follower provides a bid $p_{1-i}^j \in \{0,1,\cdots\}$.
    \item If $p_i^j \leq p_{1-i}^j$, $N_{1-i}$ successfully purchased $m_i^j$ units of $SNFT(id)$ from $N_i$ at the unit price of $\frac{p_i^j + p_{1-i}^j}2$.
    \item If $p_i^j \geq p_{1-i}^j+1$, $N_i$ purchases $m_i^j$ units of $SNFT(id)$ from $N_{1-i}$ at the unit price of $\frac{p_i^j + p_{1-i}^j}2$. And $N_{1-i}$ only obtains a discounted revenue $m_i^j\cdot \frac{p_i^j + p_{1-i}^j-1}2$.
\end{itemize}
The whole game process is shown in Figure~\ref{fig:2PRG}. Based on the description for the $j$-th round of repeated game, the utilities of $N_0$ and $N_1$ are
\begin{equation}\label{2pru0}
    U_0(m_0^j,m_1^j,p_0^j,p_1^j)=
    \begin{cases}
    (v_0 - (p_0^j + p_1^j)/2)m_1^j & if\; p_0 ^j\geq p_1^j, m_0^j > m_1^j; \\
    ( (p_0^j + p_1^j - 1)/2 - v_0)m_1^j & if\; p_0^j < p_1^j, m_0^j > m_1^j;\\
    ( (p_0^j + p_1^j)/2  - v_0)m_0^j & if \;p_1^j \geq p_0^j, m_0^j<m_1^j; \\
    (v_0 - (p_1^j + p_0^j)/2)m_0^j& if \;p_1^j < p_0^j,m_0^j<m_1^j;
    \end{cases}
\end{equation}
\begin{equation}\label{2pru1}
    U_1(m_0^j,m_1^j,p_0^j,p_1^j)=
    \begin{cases}
    (  (p_0^j + p_1^j)/2  - v_1)m_1^j& if \;p_0^j \geq p_1^j,m_0^j>m_1^j; \\
    (v_1 - (p_0^j + p_1^j)/2)m_1^j& if \;p_0^j < p_1^j,m_0^j>m_1^j;\\
    (v_1 - (p_0^j + p_1^j)/2)m_0^j& if \;p_1^j \geq p_0^j, m_0^j<m_1^j; \\
    ( (p_0^j + p_1^j-1)/2  - v_1) m_0^j& if \;p_1^j < p_0^j,m_0^j<m_1^j.
    \end{cases}
\end{equation}
Both participants are interested in their total utilities in the whole process
\begin{equation*}
  U_i = \sum_{j\in \{1,2,\cdots,t\}} U_i(m_0^j,m_1^j,p_0^j,p_1^j).  
\end{equation*}

\begin{figure}[t]
    \centering
    \includegraphics[width=5in]{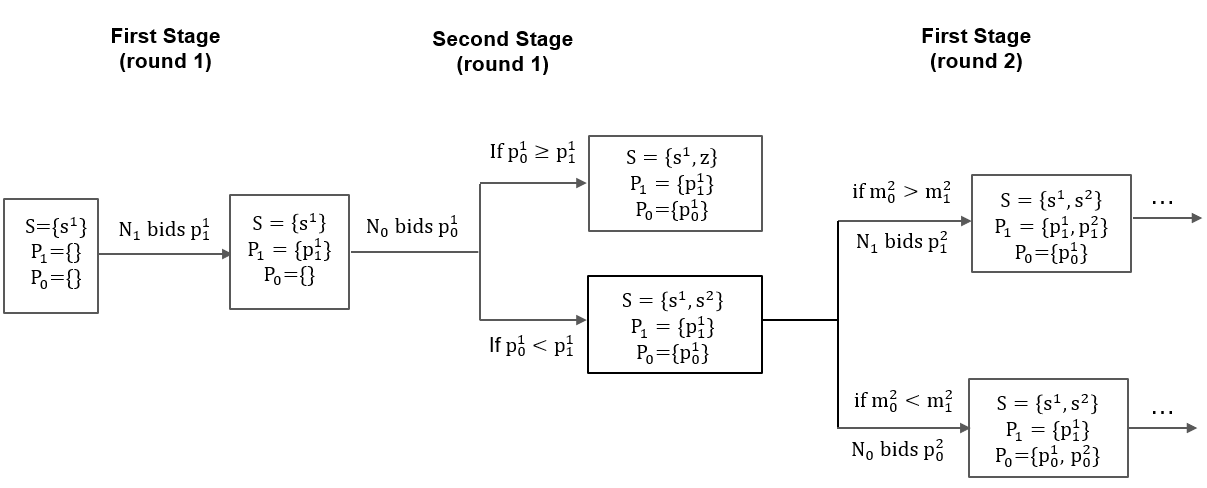}
    \caption{Two-player repeated repurchase Stackelberg Game.}
    \label{fig:2PRG}
\end{figure}

\begin{lemma}\label{lemma:tru}
For each participant $N_i$, $i\in \{0,1\}$, if his bid is set as $p_i^j = v_i$ in the $j$-th round, $j\in \{1,2,\cdots,t\}$, then $U_i(m_0^j,m_1^j,p_0^j,p_1^j)\geq 0$.
\end{lemma}
Lemma~\ref{lemma:tru} can be directly deduced from (\ref{2pru0}) and (\ref{2pru1}).

\begin{lemma}\label{lemma:forever}
If the repeated game goes through indefinitely, that is $t = +\infty$,
then $U_0 + U_1 = -\infty$.
\end{lemma}
\begin{proof} For the $j$-th round, let $N_i=l^j$ be the leader and thus $N_{1-i}$ is the follower.
Since there are only two players, all $SNFT(id)$ will belong to one player, if the follower can successfully repurchase $SNFT(id)$ from the leader, and then the repeated game stops. It means that in the $j$-th round, $m_i^j$ units of $SNFT(id)$ is bought by $N_{1-i}$ from $N_i$ and the game stops at the terminal state $z_{1-i}$. So if the repeated game goes through indefinitely, 
it must be that in each $j\in \{1,2,\cdots\}$, $p_i^j>p_{1-i}^j$, and $N_i$ buys $m_i^j$ from $N_{1-i}$. Thus in the $j+1$-th round, $m_i^{j+1}=2m_i^j$.
\begin{eqnarray}
  &&U_0(m_0^j,m_1^j,p_0^j,p_1^j)+U_1(m_0^j,m_1^j,p_0^j,p_1^j)=\nonumber\\
   &&\begin{cases}
    (v_0 - v_1)m_0^j-\frac12m_0^j=(v_0-v_1)(m_0^{j+1}-m_0^j)-\frac12m_0^j & \text{if } N_0 \text{ is the leader}; \\
    (v_1 - v_0)m_1^j-\frac12m_1^j=(v_1-v_0)(m_1^{j+1}-m_1^j)-\frac12m_1^j & \text{if } N_1 \text{ is the leader};
    \end{cases}\nonumber\\
   &&\leq (m_0^{j+1} - m_0^j) (v_0 -v_1) - \frac{1}{2};\label{sum}
 \end{eqnarray}
 and
\begin{eqnarray*}
    U_0 + U_1 &=&\lim_{t\rightarrow +\infty}\sum_{j = \{1,2,\cdots,t\}} U_0(m_0^j,m_1^j,p_0^j,p_1^j) +  U_1(m_0^j,m_1^j,p_0^j,p_1^j)\\
    &\leq &  \lim_{t\rightarrow +\infty}\sum_{j = \{1,2,\cdots,t\}} \left[(m_0^{j+1} - m_0^j)(v_0 -v_1) - \frac{1}{2}\right]\\
    &=& \lim_{t\rightarrow +\infty}\left[(m_0^{t+1} - m_0^1)(v_0 -v_1) - \frac{1}{2}t\right] \leq M|v_0 -v_1| - \lim_{t\rightarrow +\infty}\frac{1}{2}t   = - \infty.
\end{eqnarray*}
This result holds. \qed
\end{proof}

Combining Lemma~\ref{lemma:tru} and Lemma~\ref{lemma:forever}, we have the following conclusion.

\begin{lemma}\label{lemma:nonneg}
If there is a Stackelberg equilibrium in the two-player repeated Stackelberg game, then $U_0 + U_1 \geq 0$ in this Stackelberg equilibrium.
\end{lemma}
\begin{proof}
Suppose to the contrary that $U_0 + U_1< 0$ in this Stackelberg equilibrium, then there must exist $i\in \{0,1\}$, such that $U_i<0$. However, by Lemma~\ref{lemma:tru}, we know that if each player sets its price as $p_i^j=v_i$, then its utility $U_i^j\geq 0$. Hence $N_i$ can obtain more utility by setting $p_i^j=v_i$,
which is a contradiction that $N_i$ doesn't give the best response in this Stackelberg equilibrium. \qed
\end{proof}

\begin{lemma}\label{lemma:finite}
If there is a Stackelberg equilibrium in the two-player repeated Stackelberg game, then the repeated game stops in a finite number of steps, meaning $t < +\infty$, in this Stackelberg equilibrium.
\end{lemma}

The following theorem states that once a Stackelberg equilibrium exists and  $v_i>v_{1-i}$, then this player $N_i$ must buy all $SNFT(id)$ at last.
\begin{theorem}\label{theorem:maxwellfare}
If $v_i > v_{1-i}$, $i=0,~1$, and a Stackelberg equilibrium exists, then $z = z_{i}$, in all Stackelberg equilibria.
\end{theorem}
\begin{proof}
By (\ref{2pru0}) and (\ref{2pru1}), we have
\begin{equation*}
   U_0(m_0^j,m_1^j,p_0^j,p_1^j) + U_1(m_0^j,m_1^j,p_0^j,p_1^j) \leq (m_0^{j+1} - m_0^j) (v_0 -v_1). 
\end{equation*}
\begin{eqnarray*}
    U_0 + U_1 &\leq&\sum_{j \in \{1,2,\cdots,t\}} U_0(m_0^j,m_1^j,p_0^j,p_1^j) +  U_1(m_0^j,m_1^j,p_0^j,p_1^j)\\
    &\leq &  \sum_{j \in \{1,2,\cdots,t\}} (m_0^{j+1} - m_0^j)(v_0 -v_1)= (m_0^{t+1} - m_0^1)(v_0 -v_1).
\end{eqnarray*}
If $v_0 > v_1$, then it must be $m_0^{t+1} > m_0^1$. Otherwise, $U_0+U_1<0$, showing no Stackelberg equilibrium exists.
This is a contradiction. 
Because the  repeated  game  stops  in  a  finite  number  of  steps, $m_0^{t+1} \in \{0,M\}$. Combing the condition $m_0^{t+1}>m_0^1>0$, we have $m_0^{t+1}=M$. Therefore, at last $z = z_0$.
Similarly, it is easy to deduce $z = z_1$ if $v_1 > v_0$. \qed
\end{proof}

Based on Theorem \ref{theorem:maxwellfare}, we go to prove the existence of the Stackelberg equilibrium by proposing an equilibrium strategy in the following theorem.

\begin{theorem}\label{theorem:2PRSE}
If $v_i  > v_{1-i}$, $i=0,~1$, the following strategy is a Stackelberg equilibrium:
\begin{equation}\label{SEstrategy}
 p_{1-i}^j = v_{1-i};~~   p_i^j=
    \begin{cases}
    v_{1-i} + 1& if \;l^j = i;\\
    p^j_{1-i} & if \;l^j = 1-i, ~p^j_{1-i} \leq v_{1-i};\\
    p^j_{1-i} - 1& if \;l^j = 1-i, ~p^j_{1-i} > v_{1-i}.
    \end{cases}
\end{equation}
\end{theorem}
The proof of Theorem~\ref{theorem:2PRSE} is provided in Appendix B.

\section{Multi-Player Repurchase Stackelberg Game}\label{MPGame}

This section goes to extend the discussion for the multi-player scenario, in which $N_0$ has more than half of $SNFT(id)$, and $\{N_1,\cdots,N_k\}$ are repurchased participants. $N_0$ triggers the repurchase process, and asks all other repurchased participants to report their bids $p_i$ at first, and $N_0$ decides his bid $p_0$ later. We also model the repurchase process of the multi-player scenario as a two-stage Stackelberg game, where $\{N_1,\cdots,N_k\}$ are the leaders to determine their bids in Stage I, and $N_0$ acts as the followers to decide his bid $p_0$ in Stage II. Different from the two-player scenario, $N_0$ shall trade with each $N_i$, $i=1,\cdots,k$. Then each $N_i$, $i=1,\cdots,k$, has his utility $U_i(p_0,p_i)$ as (\ref{uN0}) and (\ref{uNi}). But the utility of $N_0$ is the total utility from the trading with all $N_i$. That is
\begin{eqnarray*}
  U_0(p_0,p_1,\cdots,p_k)=\sum_{i=1}^kU_0^i(p_0,p_i),
\end{eqnarray*}
where $U_0^i(p_0,p_1)$ is defined as (\ref{uN0}) and (\ref{uNi}).

\subsection{Analysis of Stackelberg Equilibrium}

In the Stackelberg repurchase game for multi-player scenario, $N_0$ shall trade with each $N_i$, $i=1,\cdots,k$. Inspired by the Stackelberg equilibrium in the two-player Stackelberg game, we first discuss the best response of $N_0$, if each $N_i$ reports his bid as
\begin{equation}\label{MPU0}
    p_i^{*}=
    \begin{cases}
    v_0  & if \;v_i \leq v_0; \\
    v_0 + 1  & if \;v_i \geq v_0+1.
    \end{cases}
\end{equation}
Then we study the collusion from a group of repurchased players. Our task is to prove that once a group of repurchased participants deviate from the bidding strategy (\ref{MPU0}), then their total utility must be decreased. This guarantees that each participant would like to follow the bidding strategy (\ref{MPU0}).

\begin{lemma}\label{lemma:MPSE}
In the Stackelberg repurchase game for the multi-player scenario, if all leaders set their bids $\{p_i^*\}$ as (\ref{MPU0}) in Stage I, then the best response of the follower $N_0$ in Stage II is $p_0^*(p_1^{*},\cdots,p_n^{*}) = v_0$.
\end{lemma}
\begin{proof}
For each trading between $N_0$ and $N_i$, $i=1,\cdots,k$, Lemma \ref{lemma1} ensures that $v_0=argmax_{p_0}U_0^i(p_0,p_i^*)$. Since each $U_0^i(p_0,p_i^*)\geq 0$, we have
\begin{eqnarray*}
  p_0^*(p_1^*,\cdots,p_k^*)=\argmax_{p_0}U_0(p_0,p_1^*,\cdots,p_k^*)=\argmax_{p_0}\sum_{i=1}^kU_0^i(p_0,p_i^*)=v_0.
\end{eqnarray*}
This lemma holds.\qed
\end{proof}

To study the collusion of repurchased participants, we partition the set of $\{N_1,\cdots,N_k\}$ into two disjoint subsets $A$ and $B$, such that each $N_i\in A$ follows the bidding strategy (\ref{MPU0}), while each $N_i\in B$ does not. Thus given all bids provided by players, the bid profile $\mathrm{\mathbf{p}}=(p_0, \{p_i^*\}_{N_i\in A}, \{p_i\}_{N_i\in B})$ can be equivalently expressed as $\mathrm{\mathbf{p}}=(p_0,\mathrm{\mathbf{p}}^*_A,\mathrm{\mathbf{p}}_B)$.
Here we are interested in the total utility of all players in $B$, and thus define
\begin{eqnarray*}
  U_B(p_0,\mathrm{\mathbf{p}}^*_A,\mathrm{\mathbf{p}}_B) = \sum_{N_i\in B}U_i(p_0, p_i).
\end{eqnarray*}
Following Lemma shows that once a group of participants deviate from the bidding strategy (\ref{MPU0}), then their total utility decreases.
\begin{lemma}\label{lemma:nonco}
Let $A=\{N_i|p_i=p_i^*\}$ and $B=\{N_i|p_i\neq p_i^*\}$. Then
\begin{eqnarray*}
  U_B(p_0^*(\mathrm{\mathbf{p}}^*_A,\mathrm{\mathbf{p}}_B), \mathrm{\mathbf{p}}^*_A,\mathrm{\mathbf{p}}_B) < U_B(v_0, p_1^*,p_2^*,\cdots,p_k^*).
\end{eqnarray*}
\end{lemma}
The proof of Theorem~\ref{lemma:nonco} is provided in Appendix C.

\begin{theorem}\label{theorem:MSE}
In the multi-player Stackelberg repurchase game, the bid profile $(v_0,p_1^*\cdots,p_k^*)$ is a Stackelberg equilibrium, where $p_i^*$ is set as (\ref{MPU0}).
\end{theorem}
\begin{proof}
To simplify our discussion, we define the price profile $\mathrm{\mathbf{p}}^*=(p_1^*,\cdots,p_k^*)$, and $\mathrm{\mathbf{p}}_{-i}^*$ denotes the profile without the price of $N_i$. So $\mathrm{\mathbf{p}}^*=(\mathrm{\mathbf{p}}^*_{-i},p^*_i)$. From Lemma~\ref{lemma:MPSE}, we have the best response of $N_0$ in Stage II is $p_0^*(\mathrm{\mathbf{p}}^*)=v_0$. However, Lemma \ref{lemma:nonco} indicates that no one would like to deviate from the pricing strategy (\ref{MPU0}), as
$
  U_i(p_0^*(\mathrm{\mathbf{p}}^*_{-i},p_i),\mathrm{\mathbf{p}}^*_{-i},p_i)<U_i(v_0,\mathrm{\mathbf{p}}^*_{i}).
$
Thus given the price profile $\mathrm{\mathbf{p}}^*$, nobody would like to change its strategy $p_i^*$ unilaterally. Therefore, $(v_0,p_1^*\cdots,p_k^*)$  is a Stackelberg equilibrium. \qed
\end{proof}

From the perspective of cooperation, we can observe that no group of repurchased participants would like to collude to deviate from the bidding strategy (\ref{MPU0}) by Lemma \ref{lemma:nonco}. Thus we have the following corollary.
\begin{corollary}
Given the Stackelberg equilibrium of  $(v_0,p_1^*\cdots,p_k^*)$, no group of repurchased participants would like to deviate this equilibrium.
\end{corollary}

In the case of incomplete information, the analysis of the Bayesian Stackelberg equilibrium becomes extremely complicated. As discussed in Section~\ref{2PBSE}, in the case of the two-player Stackelberg game, the leader only needs to optimize the utility based on incomplete information. However, when there are multiple leaders, the strategies of leaders should reach a Bayesian Nash equilibrium, which is much more difficult to calculate. So we regard it as our future work to analyze the Bayesian Stackelberg equilibrium of the multi-player repurchase Stackelberg game.

\section{Discussion}\label{discussion}

\subsection{A Blockchain Solution to Budget Constraints}\label{Budget}
In the previous settings, we do not consider the budget constraints. However, this is a common problem for many newly proposed mechanisms. Therefore, we propose a solution scheme by blockchain for the setting with budget constraints.

Suppose $N_0$ owes more than half of $SNFT(id)$ and triggers the repurchase process. Our mechanism consists of two stages. All participants except for $N_0$ report their bids in Stage I, and $N_0$ gives his bid $p_0$ in Stage II. We assume $N_0$'s budget is larger than $(M- m_0)p_0$, so that he can repurchase all other shares at his bid $p_0$. 
For $N_i$, $i \neq 0$, if $p_i > p_0$, $N_i$ should pay $\frac{p_0+p_i}{2}m_i$. However, the payment of $\frac{p_0+p_i}{2}m_i$ may exceed his budget, such that $N_i$ has not enough money to buy $m_i$ units of $SNFT(id)$. Under this situation, we provide a blockchain solution for $N_i$ to solve the problem of budget shortage. That is, we allow $N_i$ to sell his option of buying $m_i$ units of $SNFT(id)$ to anyone in the blockchain system. If nobody would like to buy $N_i$'s repurchase option, then $N_0$ can repurchase $N_i$'s shares at a lower price.
Therefore, after reporting bids, additional four steps are needed to finish the payment procedure. 

\begin{itemize}
\item Step 1. $N_0$ pays $\sum_{i\in\{1,2,\cdots,k\},p_i \leq p_0}\frac{p_0+p_i}{2}m_i$. After the payment, $N_0$ gets $\sum_{i\in\{1,2,\cdots,k\},p_i \leq p_0}m_i$ pieces of $SNFT(id)$.
    For each $N_i$ with $p_i\leq p_0$, $i\in\{1,2,\cdots,k\}$, he gets the revenue of $\frac{p_0+p_i}{2}\cdot m_i$ and loses $m_i$ units of $SNFT(id)$.
    
\item Step 2. For all $i \in \{1,2,\cdots,k\}$ that $p_i > p_0$, $N_i$ shall pay $\frac{p_0 + p_i}{2}\cdot m_i$ to buy $m_i$ units of $SNFT(id)$ from $N_0$. Once $m_i$ units of $SNFT(id)$ of $N_0$ is sold to $N_i$, $N_0$ obtains a discounted revenue $\frac{p_0 + p_i-1}{2}\cdot m_i$.\\
 If $N_i$ would not like to repurchase $SNFT(id)$, then he can
sell his repurchase option to others at a price of $\widetilde{p}_i\in \mathbb{Z}$. The price of repurchase option $\widetilde{p}_i$ could be negative, meaning that $N_i$ shall pay $\widetilde{p}_i$ to another who accepts his chance.
If $N_i$ does nothing, we regard that $N_i$ proposes $\widetilde{p}_i= 0$. 
    
\item Step 3. If a participant in the blockchain system accepts the price of $\widetilde{p}_i$, then he would propose a transaction to buy $m_i$ units of $SNFT(id)$ from $N_0$. The total cost of this participant is $\widetilde{p}_i+\frac{p_0+p_i}{2}\cdot m_i$, in which $\widetilde{p}_i$ is paid to $N_i$ and $N_0$ obtains a discounted revenue of $\frac{p_0+p_i-1}{2}\cdot m_i$. And $m_i$ units of $SNFT(id)$ are transferred from $N_0$ to the participant who buys the repurchase option.\\
At the end of this step, let $C$ be the participant set, in which each participant's repurchase option hasn't been sold yet. 

\item Step 4. For each participant $N_i\in C$, $N_0$ repurchases $m_i$ units of $SNFT(id)$ from $N_i$
at a lower price of $2p_0-p_i (<p_0)$. At the end of this step, $N_0$ obtains $m_i$ units of $SNFT(id)$, and $N_i$ obtains a revenue of $(2p_0-p_i)\cdot m_i$.   
   
\end{itemize}

\subsection{A Blockchain Solution to Lazy Bidders}

Under some circumstances, an $SNFT(id)$ holder might not bid in the repurchase process, who is named as a lazy bidder. This lazy behavior may block the repurchase process. To solve the problem caused by lazy bidders, we propose the following two schemes.

\begin{itemize}
    \item\textbf{Custody Bidding}. NFT's smart contract supports the feature for the $SNFT(id)$ holders to assign administrators to report a bid when the holder is idle or fails to make a bid. 

    \item \textbf{Value Predetermination}. Whenever a participant obtains any units of $SNFT(id)$, this participant is required to predetermine the value at which he is willing to bid, and this information is stored in the smart contract. At the beginning of the repurchase process, if a participant fails to make a bid within a certain amount of time, the smart contract automatically reports this participant's predetermined bid. This does not mean, however, that the participant has to bid at the predetermined price if he decides to make an active bid. 
\end{itemize}

\section{Conclusion}
In this paper, we propose a novel securitization and repurchase scheme for NFT to overcome the restrictions in existing NFT markets. We model the NFT repurchase process as a Stackelberg game and analyze the Stackelberg equilibria under several scenarios. To be specific, in the setting of the two-player one-round game, we prove that in a Stackelberg equilibrium, $N_0$, the participant who triggers the repurchase process, shall give the bid equally to his own value estimate. In the two-player repeated game, all securities shall be finally owned by the participant who has a higher value estimate. In the setting of multiple players, cooperation among participants cannot bring  higher utilities to them. What's more, each participant can get non-negative utility if he bids truthfully in our repurchase process.

How to securitize and repurchase NFT efficiently is a popular topic in the field of blockchain. Our work proposes a sound solution for this problem. In the future, we continue to refine our theoretical analysis. First, for the multi-player repurchase Stackelberg game, we will consider the case with incomplete information and explore the Bayesian Stackelberg equilibrium.
Second, a model of blockchain economics will be constructed to analyze the payment procedure in Section~\ref{Budget}. Furthermore, there exist some other interesting problems, including how to securitize and repurchase a common-valued NFT~\cite{kagel2009common}, how to host Complete NFTs or Securitized NFTs in decentralized custody protocols~\cite{chen2020decentralized}, whether ABSNFT can serve as a price Oracle, and so on.

\section*{Acknowledgment }

This research was partially supported by the National Major Science and Technology Projects of China-“New Generation Artificial Intelligence” (No. 2018AAA\\0100901), the National Natural Science Foundation of China (No. 11871366), and Qing Lan Project of Jiangsu Province.

\bibliographystyle{splncs04}

\bibliography{references}

\appendix
\newpage
\section*{Appendix}

\subsection*{A. Proof of Theorem \ref{SEalsoNE}}
\begin{proof}
From Theorem~\ref{t1} we know that the best response of $N_0$ is always $p_0^*=v_0$. Next, we shall discuss the best response of $N_1$ under the condition that $N_0$'s bidding strategy is $p_0=v_0$. By (\ref{Utility_of_N1}), we have
{\small\begin{equation*}
    U_1(v_0,p_1) =
    \begin{cases}
    m_1 (\frac{v_0 + p_1}{2} - v_1)  & if \;p_1 \leq v_0; \\
    m_1 (v_1 - \frac{v_0 + p_1}{2}) & if \;p_1 \geq  v_0+1.
    \end{cases}
\end{equation*}}
So $U_1$ monotonically increases when $p_1 \leq v_0$ and monotonically decreases when $p_1 \geq v_0+1$, implying $p_1^*\in\{v_0,v_0+1\}$. Particularly, when $v_0\geq v_1$, we have
{\small\begin{eqnarray*}
  U_1(v_0,v_0)=m_1(v_0-v_1)\geq 0> m_1(v_1-v_0-\frac12)=U_1(v_0,v_0+1),
\end{eqnarray*}}
showing the best response of $N_1$ is $p_1^*=v_0$. On the other hand, when $v_0\leq v_1-1$,
{\small\begin{eqnarray*}
  U_1(v_0,v_0)=m_1(v_0-v_1)< 0< m_1(v_1-v_0-\frac12)=U_1(v_0,v_0+1),
\end{eqnarray*}}
showing the best response of $N_1$ is $p_1^*=v_0+1$. This result holds. \qed
\end{proof}

\subsection*{B. Proof of Theorem \ref{theorem:2PRSE}} 
\begin{proof}
Let us denote the above strategy (\ref{SEstrategy}) as $p_i^*(j)$ and $p_{1-i}^*(j)$. We have
\begin{equation}\label{2PRSEU1}
    U_{1-i}(s^j,p_i^*(j),p_{1-i}^j)=
    \begin{cases}
    (v_{1-i}-(p_{1-i}^j + v_{1-i} + 1)/2)m_i& if \;l^j = i, p_{1-i}^j \geq v_{1-i} + 1\\
    ((p_{1-i}^j + v_{1-i})/2 - v_{1-i})m_i& if \;l^j = i, p_{1-i}^j \leq v_{1-i}\\
    (p_{1-i}^j - v_{1-i})m_{1-i} & if \;l^j = 1-i, p_{1-i}^j \leq v_{1-i}\\
    (v_{1-i} - (2p_{1-i}^j - 1)/2)m_{1-i}& if \;l^j = 1-i, p_{1-i}^j > v_{1-i}
    \end{cases}
\end{equation}
Denote $BR^{1-i}(p_i^*(j))$ to be the best responses of $N_{1-i}$ with respect to strategy $p_i^*(j)$, and $BR^{i}(p_{1-i}^*(j))$ is similar. Then
$$BR^{1-i}(p_i^*(j)) = \argmax_{p_{1-i}^1,p_{1-i}^2, \cdots, p_{1-i}^t \in \{0,1,2,\cdots\}} U_{1-i}$$
Obviously, $U_{1-i}(s^j,p_i^*(j),p_{1-i}^*(j)) = 0$. If $p_{1-i}^j \neq p_{1-i}^*(j)$, from (\ref{2PRSEU1}) we have
$U_{1-i}(s^j,p_i^*(j),p_{1-i}^j) < 0 = u_{1-i}(s^j,p_i^*(j),p_{1-i}^*(j)).$
So
$$
   BR^{1-i}(p_i^*(j)) = \argmax_{p_{1-i}^1,p_{1-i}^2, \cdots, p_{1-i}^t \in \{0,1,2,\cdots\}} U_{1-i} = p_{1-i}^*(j).
$$

Now let us consider the case that $p_{1-i}^j = p_{1-i}^*(j)$, for any $j \in \{0,1,2,\cdots,t\}$.

Denote $U_i^* = \sum_{j =1}^{t} U_i(m_0^j,m_1^j,p_i^*(j),p_{1-i}^*(j))$, and we have $U_i^* > 0$.

(1) If $m_i^d < m_{1-i}^d$, then $l^d = i$.

If $p_i^d \leq p_{1-i}^*(d)$, then $z = z_{1-i}$. And we have $U_i + U_{1-i} <0$ when $z = z_{1-i}$. Combined with $U_{1-i} \geq 0$, we have $U_i < U_i^*$.

If $p_i^d > p_{1-i}^*(d)$, then $$\argmax_{p_i^d > p_{1-i}^*(d)}U_i = \argmax u_i(s^d,p_i^*(d),p_{1-i}^*(d)) = p_{1-i}^*(d) + 1 = p_i^*(d)$$

So $\argmax_{p_i^d \in \{0,1,2,\cdots\}}U_i  = p_i^*(d)$.

(2) When $m_i^d > m_{1-i}^d$, we have $l^d = 1 - i$.

For $ 1\leq d \leq t$, we define $U_i(d) = \sum_{j = \{d,d+1,\cdots,t\}}u_i(m_0^j,m_1^j,p_0^j,p_1^j)$. Similarly, we define $U_{1-i}(d) = \sum_{j = \{d,d+1,\cdots,t\}}u_i(m_0^j,m_1^j,p_0^j,p_1^j)$.

If $p_i^d \neq p_i^*(d)$ $u_{1-i}(s^d,p_i^d,p_{1-i}^*(d)) > 0$, so we have $U_{1-i}(d) > 0$ when $p_i^d \neq p_i^*(d)$.
Then
\begin{eqnarray*}
    U_i(d) + U_{1-i}(d) &=&\sum_{j \in \{d,d+1,\cdots,t\}} u_0(m_0^j,m_1^j,p_0^j,p_1^j) +  u_1(m_0^j,m_1^j,p_0^j,p_1^j)\\
    &\leq &  \sum_{j \in \{d,d+1,\cdots,t\}} (m_i^{j+1} - m_i^j)(v_i -v_{1-i})\\
    &=& (m_i^{t+1} - m_i^d)(v_i -v_{1-i})
\end{eqnarray*}

If $p_i^d = p_i^*(d)$, $U_i(d) = (m_i^{t+1} - m_i^d)(v_i -v_{1-i})$. If $p_i^d \neq p_i^*(d)$, $U_{1-i}(d) >0$, then $U_i(d) < (m_i^{t+1} - m_i^d)(v_i -v_{1-i})$. So
$$\argmax_{p_i^d\in\{0,1,2,\cdots\}}U_i = \argmax_{p_i^d\in\{0,1,2,\cdots\}, }U_i(d) = p_i^*(d).$$

Above all, we have
\begin{eqnarray}
    BR^i(p_{1-i}^*(j)) = \argmax_{p_i^j\in\{0,1,2,\cdots\}, j\in \{0,1,2,\cdots, t\}}U_i = p_i^*(j)
\end{eqnarray} \qed
\end{proof}

\subsection*{C. Proof of Lemma \ref{lemma:nonco}}

\begin{proof}
By Lemma \ref{lemma2}, if follower $N_0$ sets its price as $p_0=v_0$ in Stage II, then the optimal price is $p_i=p_i^*$ by leader $N_i$ is Stage I. It means that $0\leq U_i(v_0,p_i) < U_i(v_0,p_i^*)$, implying
\begin{equation}\label{CE1}
    \sum_{N_i\in B}U_i(v_0, p_i) < \sum_{N_i\in B}U_i(v_0, p_i^*).
\end{equation}
Let us simplify the best response $BR(p_1,\cdots,p_k)$ of $N_0$ as $BR$. From Lemma~\ref{lemma1} we have
\begin{equation}\label{CE5}
    \sum_{N_i\in A}U_0^i(BR_2, p_i^*) \leq \sum_{N_i\in A}U_0^i(v_0, p_i^*).
\end{equation}
From the utility function (\ref{uN0}) and (\ref{uNi}), we have
\begin{equation}\label{U0PlusUI}
U_0^i(p_0, p_i) + U_i(p_0, p_i) =
    \begin{cases}
        m_i(v_0 - v_i)  & if \;p_0 \geq p_i; \\
        m_i(v_i - v_0) - \frac{1}{2}  & if \;p_0 \leq p_i - 1.
    \end{cases}
\end{equation}
Then
\begin{equation}\label{U0PlusUIMAX}
\max_{p_0,p_i}(U_0^i(p_0, p_i) + U_i(p_0, p_i)) =
    \begin{cases}
        m_i(v_0 - v_i)  & if \;v_0 \geq v_i; \\
        m_i(v_i - v_0) - \frac{1}{2}  & if \;v_0 \leq v_i - 1.
    \end{cases}
\end{equation}
$$\sum_{N_i\in B}U_0^i(p_0, p_i) + \sum_{N_i\in B}U_i(p_0,p_i) \leq \sum_{N_i\in B,v_i \leq v_0}m_i(v_0-v_i) + \sum_{N_i\in B,v_i \geq v_0+1}(m_i(v_i-v_0) - \frac{1}{2}).$$
Denote $C_B =  \sum_{N_i\in B,v_i \leq v_0}m_i(v_0-v_i) + \sum_{N_i\in B,v_i \geq v_0+1}(m_i(v_i-v_0) - \frac{1}{2})$, then
\begin{equation}\label{CE2}
\sum_{N_i\in B}U_0^i(p_0, p_i) + U_B(p_0, \mathrm{\mathbf{p}}_A^*,\mathrm{\mathbf{p}}_B) \leq C_B,
\end{equation}
and
\begin{equation}\label{CE3}
\sum_{N_i\in B}U_0^i(v_0, p_i^*) + U_B(v_0, p_1^{*}, \cdots p_k^{*}) = C_B.
\end{equation}
Because $BR_2(\mathrm{\mathbf{p}}_A^*,\mathrm{\mathbf{p}}_B)$ is the best response of $N_0$ in Stage II given other players' prices $(\mathrm{\mathbf{p}}_A^*,\mathrm{\mathbf{p}}_B)$,
$$U_{0}(BR_2, \mathrm{\mathbf{p}}_A^*,\mathrm{\mathbf{p}}_B) \geq U_{0}(v_0, \mathrm{\mathbf{p}}_A^*,\mathrm{\mathbf{p}}_B).$$
In addition, we have
\begin{eqnarray*}
U_{0}(BR_2, \mathrm{\mathbf{p}}_A^*,\mathrm{\mathbf{p}}_B) &=& \sum_{N_i\in A}U_0^i(BR_2,p_i^*) + \sum_{N_i\in B}U_0^i(BR_2,p_i)\\
& \leq &\sum_{N_i\in A} U_0^i(v_0, p_i^*) + \sum_{N_i\in B}U_0^i(BR_2,p_i),
\end{eqnarray*}
where the inequality is from (\ref{CE5}). So
\begin{eqnarray}
  \sum_{N_i\in B}U_0^i(BR_2,p_i) &\geq & U_{0}(BR_2, \mathrm{\mathbf{p}}_A^*,\mathrm{\mathbf{p}}_B) - \sum_{N_i\in A} U_0^i(v_0, p_i^*)\nonumber\\
  &\geq &U_{0}(v_0, \mathrm{\mathbf{p}}_A^*,\mathrm{\mathbf{p}}_B) - \sum_{N_i\in A} U_0^i(v_0, p_i^*)=\sum_{N_i\in B} U_0^i(v_0, p_i).\label{CE4}
\end{eqnarray}

Combining (\ref{CE2}), (\ref{CE3}) and (\ref{CE4}), we have
\begin{eqnarray*}
U_B(BR_2,\mathrm{\mathbf{p}}_A^*,\mathrm{\mathbf{p}}_B)  &\leq &C_B - \sum_{N_i\in B}U_0^i(BR_2, p_i) \\
&\leq&  C_B - \sum_{N_i\in B}U_0^i(v_0,p_i)\\
&=&\sum_{N_i\in B}U_0^i(v_0, p_i^*) + U_B(v_0, p_1^{*}, \cdots p_k^{*})- \sum_{N_i\in B}U_0^i(v_0,p_i)\\
&=&U_B(v_0, p_1^{*}, \cdots p_k^{*}) + (\sum_{N_i\in B}U_0^i(v_0, p_i^*)- \sum_{N_i\in B}U_0^i(v_0,p_i))\\
&<&U_B(v_0, p_1^{*}, \cdots p_k^{*}),
\end{eqnarray*}
where the last inequality is from (\ref{CE1}). This lemma holds. \qed
\end{proof}

\end{document}